\documentclass[aps,prx, showpacs, twocolumn, superscriptaddress]{revtex4-1}
\usepackage{graphicx}  
\usepackage{dcolumn}   
\usepackage{bm}        
\usepackage{amssymb}   
\usepackage{amsmath}   
\usepackage{bbold}
\usepackage{dsfont} 
\usepackage[position=top, caption=false]{subfig} 
\usepackage{floatrow}
\usepackage{url}
\usepackage[version=3]{mhchem} 
\usepackage{algorithm}
\usepackage[noend]{algpseudocode}
\usepackage{amsthm}
\newtheorem{theorem}{Theorem}

\usepackage{hhline}

\begin{document}

\title{Taming Convergence in the Determinant Approach for X-Ray Excitation Spectra}

\author{Yufeng Liang}
\email{yufengliang@lbl.gov, pcbee912@gmail.com}
\affiliation{The Molecular Foundry, Lawrence Berkeley National Laboratory, Berkeley, CA 94720, USA}
\author{David Prendergast}
\affiliation{The Molecular Foundry, Lawrence Berkeley National Laboratory, Berkeley, CA 94720, USA}

\begin{abstract}
A determinant formalism in combination with \emph{ab initio} calculations was proposed recently \cite{liang2017accurate, liang2018quantum} and has paved a new way for simulating and interpreting x-ray excitation spectra in condensed-phase systems.
The new method systematically takes into account many-electron effects in the Mahan-Nozi\'eres-De Dominicis (MND) theory,
including core-level excitonic effects, the Fermi-edge singularity, shakeup excitations, and wavefunction overlap effects such as the orthogonality catastrophe, 
all within a universal framework using many-electron configurations.
A heuristic search algorithm \cite{liang2018quantum} was introduced to search for the configurations that are important for defining the x-ray spectral lineshape, instead of enumerating them in a brute-force way.
The algorithm has proven to be efficient for calculating \ce{O} $K$ edges of transition metal oxides, 
which converge at the second excitation order (denoted as $f^{(n)}$ with $n = 2$), i.e., the final-state configurations with two \emph{e-h} pairs (with one hole being the core hole).
However, it remains unknown how the determinant x-ray spectra converge for general cases and at which excitation order $n$ should one stop the determinant calculation. 
Even with the heuristic algorithm, the number of many-electron configurations still grows exponentially with the excitation order $n$.
In this work, we prove two theorems that can indicate the order of magnitude of the contribution of the $f^{(n)}$ configurations,
so that one can estimate their contribution very quickly without actually calculating their amplitudes.
The two theorems are based on the singular-value decomposition (SVD) analysis, a method that is widely used to quantify entanglement between two quantum many-body systems.
We examine the $K$ edges of several metallic systems with the determinant formalism up to $f^{(5)}$ to illustrate the usefulness of the theorems. 
\end{abstract}
\maketitle

\section{Introduction}
 
X-ray spectroscopy has become increasingly important for providing insights into many problems in materials characterization at microscopic scale \cite{de2008core, fadley2010x, singh2010advanced, ament2011resonant, van2016x}, 
especially in recent times when this area is propelled by the development of light sources enabled by free-electron lasers \cite{emma2010first, schlotter2012soft, lemke2013femtosecond}.
We consider resonant x-ray excitations, where a core electron of a specific type of atoms is promoted into orbitals localized on or near that atomic site, revealing a wealth of information of local chemical environment and electronic structure.
The near-edge part of a x-ray spectrum, i.e., a few eVs above onset, is of particular interest, where most intriguing physics and chemical processes are related to \cite{de2008core}.
The interpretation of x-ray spectra, however, is often a nontrivial task that requires accurate first-principles modeling of both atomic and electronic structure of interest and their associated spectra \cite{prendergast2006x, cabaret2010first, gougoussis2009first}.
While the structural properties of a wide range of materials can be mostly captured by density-functional theory (DFT) \cite{parr1980density, dudarev1998electron, bickelhaupt2000kohn, wang2006oxidation, hautier2012accuracy}, 
predicting x-ray excited-state spectra in a reliable and efficient manner presents a greater theoretical challenge. 

There has been a broad spectrum of theoretical approaches for simulating x-ray excitation spectra.
At one extreme, exact diagonalization \cite{chen1991electronic, hybertsen1992model, haverkort2012multiplet, kourtis2012exact, uldry2012systematic, chen2013doping, jia2016using, lu2017nonperturbative} has been applied to rigorously solve the many-electron Hamiltonian.
This method represents the most accurate solutions and is most suitable for localized x-ray excitations that occur within a few atomic orbitals, due to the exponential growth of the 
many-electron Hilbert space. 
Wannier down-folding is typically required for reducing the size of the Hamiltonian \cite{haverkort2012multiplet, uldry2012systematic} and the Coulomb interaction is usually simplified as an on-site Hubbard $U$ term.
Similar methodology is also employed in the quantum chemistry community within the configuration interaction \cite{sherrill1999configuration, tubman2016deterministic} and other post-Hartree-Fock methods \cite{bartlett2007coupled, szalay2011multiconfiguration},
in which \emph{realistic} Coulomb interactions and superposition of many-electron configurations are considered explicitly in the calculation. 
Likewise, increasing the size of the excited-state calculation is hindered by the exponential barrier associated with the size of the Hilbert space, 
although placing restrictions on the active orbital space can mitigate the problem to some extent \cite{roemelt2013excited, pinjari2014restricted, maganas2017restricted}.
Due to its rigor, this class of methods can be used for system with strong electron correlation, 
but due to its computational inefficiency, it is limited to specific cases such as small molecules \cite{roemelt2013excited, maganas2017restricted} or clusters \cite{lu2017nonperturbative} and $3d$ metal $L$ edges (dominate d by localized atomic multiplet effects) \cite{haverkort2012multiplet, jia2016using} so far.
 
At the other extreme, excited-state spectra of condensed phases can be routinely obtained using \emph{ab initio} methods based on DFT and many-body perturbation theory (MBPT), where electron correlation is treated with less rigor but higher efficiency.
There are two representative methods.
One is the Delta-self-consistent-field ($\Delta$SCF) approach \cite{nyberg1999core, taillefumier2002x, prendergast2006x, liang2017accurate} that treats the core hole as an external potential, and maps a many-electron excited state to an empty orbital. 
Many-electron response is taken into account by DFT electronic relaxation, and then the transition matrix elements are calculated using the Fermi's Golden rule.
In previous formalisms, however, only one-body orbitals are employed for final states, which does not account for proper time ordering of the many-electron processes in x-ray excitations \cite{liang2017accurate, liang2018quantum}, 
leading to possible failures in predicting the intensity of an absorption edge.
The other method is the core-level Bethe-Salpeter Equation (BSE) \cite{shirley1998ab, olovsson2009all, vinson2011bethe} within MBPT, that utilizes DFT ground state as a zero-order approximation and incorporate many-electron correlations in a \emph{perturbative} manner. 
A subset of Feynman diagrams generated by the direct and exchange kernels are included to account for \emph{e-h} interactions (excitonic effects). and hence the correct time-ordering is retained in this approach.
Formulations akin to the BSE are also adopted in some versions of time-dependent DFT (TDDFT) \cite{besley2009time, besley2010time, lopata2012linear, zhang2012core}, 
in which exchange-correlation kernels play the roles of \emph{e-h} interaction kernel, 
and real-time evolution could be involved \cite{lopata2012linear}.
There are many successful application of the $\Delta$SCF core-hole approach \cite{nyberg1999core, taillefumier2002x, prendergast2006x, pascal2014finite, velasco2014structure} and the core-level BSE \cite{vinson2011bethe, liang2017accurate, zhan2017enabling} in extended systems, and even in molecules \cite{nyberg1999core} where the Coulomb interaction is not well screened.

\subsection{Motivations and Advantages of the Determinant Formalism}
Following the philosophy of the post DFT methods, we recently proposed a determinant formalism \cite{liang2017accurate, liang2018quantum} for simulating x-ray excitation spectra based on the one-body core-hole approach.
The determinant formalism has three main advantages:
First, it provides an exact solution to all the many-electron effects considered within the Mahan-Nozi\'eres-De Dominicis (MND) model \cite{mahan1967excitons, nozieres1969singularities, ohtaka1990theory, mahan2013many}, 
which is beyond the scope of \emph{e-h} attraction in the BSE. 
These effects include the power-law edge singularity as considered by Mahan using Feynman diagrams \cite{mahan2013many},
and the many-body wavefunction overlap effect considered by Anderson in the orthogonality catastrophe \cite{anderson1967pw}.
In a less dramatic manner, the latter often manifests as \emph{shakeup effects} \cite{kim1975x, stern1983many, calandra2012k}, but not edge-rounding effects. 
The shake-up effects are the fragmentation of an absorption feature to higher energy, which would not be described by the BSE.

Secondly, the determinant formalism adopts many-electron configurations in the calculation, and hence provides a conceptually simple picture for understanding x-ray excited states.
In the x-ray final-state system, a many-body state is simply a single determinant (single reference) of the occupied Kohn-Sham orbitals. 
Each excited state is now mapped to its composite orbitals.
We denote the group of singles as $f^{(1)}$, which comprises a core hole and an electron; the groups of doubles as $f^{(2)}$, which comprises an \emph{additional} valence \emph{e-h} pair; and so forth.
The $n$ as in $f^{(n)}$ can be understood as the order of excitations.
\emph{This enables a straightforward assignment of absorption features to excited states,  making the interpretation of the spectrum simple.}

Thirdly, an alternative solution for solving the MND Model using DFT orbitals as input was introduced in Ref. \cite{wessely2005ab, wessely2006dynamical, wessely2007dynamical, ovcharenko2013specific}, and good agreement has been achieved in some carbon systems and metals.
In this method, a time-dependent matrix integral equation needs to be solved and matrix inversion is required for each point on the discretized time axis.
Then a Fourier transformation is performed to obtain the spectrum.
The determinant approach we proposed only requires a one-shot matrix inversion and no real-time evolution is involved.
It is physically more intuitive and computationally less complex. 

Finally, although many-electron configurations are employed for condensed-phase systems with hundreds to thousands of electrons in a supercell, 
the determinant calculation is \emph{not at all intractable}.
We have developed a heuristic search algorithm for finding the many-electron configurations that are important for determining the x-ray spectra \cite{liang2018quantum}.
For transition metal oxides (TMOs), it is found the x-ray absorption spectra (XAS) can be well defined by just $10^5$ configurations up to the $f^{(2)}$ order (not their superposition and no diagonalization of many-body Hamiltonian is needed), which are inexpensive calculations given the current computational capability.


It is, however, unknown yet how expensive the determinant calculations are in other systems apart from TMOs.
It remains unclear if $f^{(n)}$ configurations with $n > 2$ are important for shaping the x-ray excitation spectra.
Even with the heuristic search algorithm \cite{liang2018quantum} that truncates the number of many-electron configurations,
the number of meaningful $f^{(n)}$ configurations still tends to diverge exponentially with respect to the shakeup order $n$.
If one can estimate the spectral contributions from the $f^{(n)}$ configurations before actually calculating them, 
it will be of great benefit for saving computational resources.
In this work, we propose a useful criterion for estimating the contributions of $f^{(n)}$, which helps one to decide at what $n$ one should stop the determinant calculation.
It is based on an singular-value decomposition (SVD) analysis of the $\zeta$-matrix, the auxiliary matrix used to obtain the determinant transition amplitude.
An SVD analysis of the $\zeta$-matrix reflects how the one-body basis set is rotated due to the perturbation of the core hole,
and how the final-state occupied manifold is entangled with the initial state (typically the many-electron ground state).
Besides TMOs, we have chosen several metallic systems such as \ce{Li} and \ce{Cu} metal and performed exhaustive calculations up to $n = 5$ to test the convergence criterion in this work. 

\subsection{Review of the Determinant Formalism}

The central formula for calculating the x-ray absorption amplitudes from the ground-state (initial state) $|\Psi_i\rangle$ to a specific final state $|\Psi_f\rangle$ is \cite{liang2017accurate, liang2018quantum}
\begin{align}
\begin{split}
\langle \Psi_f | \mathcal{O} | \Psi_i \rangle = \sum_{c} (A^{f}_c)^* \langle \psi_c | o | \psi_h \rangle 
\end{split}
\label{eq:mat_elem}
\end{align}
in which the transition amplitude also takes a determinant form
\begin{align}
\begin{split}
A^{f}_c=
\det
\begin{bmatrix}
\xi_{f_1, 1} & \xi_{f_1, 2} & \cdots & \xi_{f_1, N} & \xi_{f_1, c} \\
\xi_{f_2, 1} & \xi_{f_2, 2} & \cdots & \xi_{f_2, N} & \xi_{f_2, c} \\
\vdots & & \ddots & & \vdots\\
\xi_{f_{N+1}, 1} & \xi_{f_{N+1}, 2} & \cdots & \xi_{f_{N+1}, N} & \xi_{f_{N+1}, c} \\
\end{bmatrix}
\end{split}
\label{eq:afc}
\end{align}
and $\xi_{ij} = \langle \psi_j | \tilde{\psi}_i \rangle$, where $|  \tilde{\psi}_i \rangle$ is one final-state orbital and $|  \psi_j \rangle$ is one initial-state orbital.
Alternatively, $A^{f}_c$ can be regarded as a \lq\lq Slater determinant\rq\rq of a set of final-state orbitals $( \tilde{\psi}_{f_1}, \tilde{\psi}_{f_2}, \cdots, \tilde{\psi}_{f_{N+1}} )$, 
expanded over the lowest $N + 1$ initial-state orbitals $( \psi_{1}, \psi_{2}, \cdots, \psi_{N+1} )$ (rather than $( \bm{r}_1, \bm{r}_2, \cdots, \bm{r}_{N + 1} )$).
The first step we use to simplify the calculation is to move the summation over $c$ into the determinant expression:
\begin{align}
\begin{split}
A^f & \equiv \langle \Psi_i | \mathcal{O} | \Psi_f \rangle = \det \bm{A}^f\\
\bm{A}^f & =  
\begin{bmatrix}
\xi_{f_1, 1} & \xi_{f_1, 2} & \cdots & \xi_{f_1, N} & \sum_c \xi_{f_1, c} w^*_c\\
\xi_{f_2, 1} & \xi_{f_2, 2} & \cdots & \xi_{f_2, N} & \sum_c \xi_{f_2, c} w^*_c \\
\vdots & & \ddots & & \vdots\\
\xi_{f_{N+1}, 1} & \xi_{f_{N+1}, 2} & \cdots & \xi_{f_{N+1}, N} & \sum_c \xi_{f_{N+1}, c} w^*_c \\
\end{bmatrix}
\end{split}
\label{eq:Af}
\end{align}
There are an enormous number of combinations of $(f_1, f_2, \cdots, f_{N + 1})$, representing possible final-state configurations,
but it is only meaningful to visit a small subset of this space.
The determinant amplitude $A^f$ is a significant number 
only when the most of $(f_1, f_2, \cdots, f_{N + 1})$ overlaps significantly with the lowest-energy configuration $(1, 2, \cdots, N + 1)$, 
because in most realistic materials, the core hole only slightly perturbs the system and the orthogonality catastrophe does not occur.
Therefore, $(f_1, f_2, \cdots, f_{N + 1})$ may only differ from $(1, 2, \cdots, N + 1)$ by a few indices.

To simplify the notation, we may denote $(f_1 = 1, f_2 = 2, \cdots, f_{N} = N, f_{N + 1} = c)$ as a single, or $f^{(1)}$ configuration,  
 $(f_1 = 1, f_2 = 2,\cdots, f_{v_1 - 1} = v_1 - 1, f_{v_1} = v_1 + 1, \cdots, f_{N - 1} = N, f_N = c, f_{N + 1} = c_1)$ as a double, or $f^{(2)}$ configuration, and so forth.

Each line of $\bm{A}^f$ can be denoted as:
\begin{align}
\begin{split}
a_i = 
\begin{bmatrix}
\xi_{i, 1} & \cdots & \xi_{i, N} &  \sum_c \xi_{i, c} w^*_c
\end{bmatrix}
\end{split}
\end{align}
Next we can introduce the $\zeta$-matrix to calculate all the $A^f$'s in an iterative manner. 
The $\zeta$-matrix is the linear transformation from $a_i$'s of the occupied orbitals  ($i \leq N$ plus $a_{N+1}$) to the $a_i$'s of the empty orbitals (i > N)
\begin{align}
\begin{split}
\begin{bmatrix}
a_{N+1} \\ a_{N+2} \\ \vdots \\ a_M
\end{bmatrix}
=
\begin{bmatrix}
0 & 0 & \cdots & 1 \\
\zeta_{N+2, 1} &  \zeta_{N+2, 2}& \cdots & \zeta_{N+2, N+1} \\
\vdots & \vdots & & \vdots \\
\zeta_{M, 1} &  \zeta_{M, 2}& \cdots & \zeta_{M, N+1}
\end{bmatrix}
\begin{bmatrix}
a_{1} \\ a_{2} \\ \vdots \\ a_{N+1}
\end{bmatrix}
\end{split}
\label{eq:zeta_def}
\end{align}
Rewriting the above matrix multiplication in a compact form, we have $\bm{A}^\text{new} = \bm{\zeta} \bm{A}^\text{ref}$,
where $(\bm{\zeta})_{ij} = \zeta_{N + i, j}$.
Thus the $\zeta$-matrix can be obtained from $\bm{\zeta} = \bm{A}^\text{new} (\bm{A}^\text{ref})^{-1} $.

With the auxiliary $\zeta$ matrix, we can quickly calculate the determinants for many excited-state configurations without repeatedly using the $\mathcal{O}(N^3)$ determinant algorithm.
Instead, an $\mathcal{O}(1)$ algorithm can be used.
Once the determinant for the ground state, $A^\text{ref}$,  is obtained, the determinant of each excited state can be computed by multiplying $A^\text{ref}$ by a pre-factor.
The pre-factor is a small determinant that composes the matrix elements of the $\zeta$ matrix.
For example, the amplitude of a single ($f^{(1)}$), double ($f^{(2)}$), and triple  ($f^{(3)}$) configuration can be obtained respectively as
\begin{align}
\begin{split}
A^c &= \zeta_{c, N} A^{\text{ref}} \\
A^{c; c_1, v_1} &= \det 
\begin{bmatrix}
\zeta_{c, v_1} & \zeta_{c, N} \\
\zeta_{c_1, v_1} &  \zeta_{c_1, N} \\
\end{bmatrix}
A^{\text{ref}}\\
A^{c; c_1, v_1;c_2,v_2} &= \det 
\begin{bmatrix}
\zeta_{c, v_2} & \zeta_{c, v_1} & \zeta_{c, N} \\
\zeta_{c_1, v_2} & \zeta_{c_1, v_1} &  \zeta_{c_1, N} \\
\zeta_{c_2, v_2} & \zeta_{c_2, v_1} &  \zeta_{c_2, N} 
\end{bmatrix}
A^{\text{ref}}
\end{split}
\label{eq:mat_elem}
\end{align}
These are the essential formulae to obtain the amplitudes for the excited-state configuration with the $\mathcal{O}(1)$ updating algorithm.
A $f^\text{(n)}$ configuration corresponds to a $n \times n $ sub-determinant of the $\zeta$ matrix. 

Directly enumerating all such sub-determinants will be a computationally expensive task.
Also, it is not necessary to do so because the $\zeta$ matrix could be a sparse matrix.
In a previous work \cite{liang2018quantum}, we have proposed a heuristic search algorithm for quickly finding significant sub-determinants of a sparse matrix.
Our general search algorithm is not merely specific to the $\zeta$ matrix for x-ray spectroscopic problems.

What was overlooked in the previous work, however, is that the $\zeta$ matrix for x-ray excitations does have structures. 
It can be seen that the $\zeta{}$ matrices for \ce{SiO2}, \ce{TiO2}, \ce{CrO2} all display a vertical (horizontal) stripe pattern (Fig. 8 (d) of Ref. \cite{liang2018quantum}).
This stripe pattern can also be seen in other chosen examples, which are discussed later in this work.
These stripe patterns imply a further simplification of the $\zeta$ matrix that it can be approximately expressed as the Kronecker product of two vectors:
\begin{align}
\begin{split}
\zeta_{ij} \sim a_i b_j
\end{split}
\label{svd}
\end{align}

If $\zeta_{ij} \sim a_i b_j$ strictly holds, then all the $n\times n$ sub-determinants for $n > 1$ will be exactly zero and only the $f^{(1)}$ amplitudes are non-vanishing.
It is the deviation of $\zeta_{ij}$ from $a_i b_j$ that determines the size of higher-order terms $f^{(n)}$ ($n > 1$).
If one can expand $\zeta_{ij}$ into just a few terms, then it is highly probable that the size of higher-order terms can be quickly estimated.
In this regard, a singular-value decomposition (SVD) of the $\zeta$ matrix is most relevant for this problem.
SVD has been widely used to analyze the entanglement structure between two quantum many-body systems \cite{vidal2007classical, li2008entanglement, schollwock2011density}.

In the following discussion, we first provide and prove two theorems that will give upper bounds on the size of the $f^{(n)}$ terms, for a specific $n$, using SVD analysis for the $\zeta$ matrix.
The bounds will enable one to determine the contribution of the $f^{(n)}$ terms to the x-ray excitation spectrum, without explicitly calculating these terms.
This will save a substantial amount of computational cost and help one obtain a meaningful x-ray excitation spectrum faster.
Then we apply the theorems to several small band-gap and metallic systems, in which higher-order terms $f^{(n)}$ ($n > 1$) are expected to contribute to the spectrum significantly.
It is, however, found that in \emph{none} of the chosen systems, the contribution from $f^{(n)}$ ($n > 2$) can significantly alter the spectral lineshapes (more precisely, the peak intensity ratios).
In other words, the spectra have already taken shape at the order of $n = 2$.

\section{Results and Discussion}

\subsection{Two theorems about sub-determinants}
\begin{theorem} 
Let $D$ be the determinant of an $n \times n $ sub-matrix that spans over rows $i_1, i_2, \cdots, i_n$ and columns $j_1, j_2, \cdots, j_n$ of an$N \times M$ matrix $\bm{\zeta}$.
Suppose the singular-value decomposition (SVD) of $\bm{\zeta}$ is
\begin{align}
\begin{split}
\zeta_{ij} = \sum_{k} s^k a^k_i b^k_j
\end{split}
\label{svd}
\end{align}
where $\{ s^k \}$ are the singular values of $\bm{\zeta}$, and $a^k_i$ ($b^k_j$) is a normalized vector for a given $k$.
Then the determinant $D$ can be expanded as the summation
\begin{align}
\begin{split}
D=\sum_{k_1 < k_2 < \cdots < k_n} s^{k_1} s^{k_2} \cdots s^{k_n} D^a_{[k_\mu]} D^b_{[k_\mu]}
\end{split}
\end{align}
\label{theorem1}
in which $D^a_{[k_\mu]}$ ($D^b_{[k_\mu]}$) is the determinant of the sub-matrix that spans over rows $i_1, i_2, \cdots, i_n$ and columns $k_1, k_2, \cdots, k_n$ of the matrix $a^k_i$ ($b^k_j$).
\end{theorem}

\begin{proof}
Without loss of generality, we may assume $i_1 = j_1 = 1$, $i_2 = j_2 = 2$, $\cdots$, $i_n = j_n = n$.
According to the definition of determinant, $D$ can be expanded using the Levi-Civita symbol:
\begin{align}
\begin{split}
D=\sum^n_{l_1 l_2 \cdots l_n = 1} \varepsilon_{l_1 l_2 \cdots l_n} \zeta_{1l_1} \zeta_{2 l_2} \cdots \zeta_{n l_n}
\end{split}
\end{align}
Inserting the SVD expression of the matrix element $\zeta_{ij}$ as in Eq. [\ref{svd}] (examples of $\bm{\zeta}$ can be found in Fig. \ref{svd}), 
\begin{align}
\begin{split}
D& =\sum^n_{l_1 l_2 \cdots l_n = 1} \varepsilon_{l_1 l_2 \cdots l_n}  \big( \sum_{k_1} s^{k_1} a^{k_1}_1 b^{k_1}_{l_1} \big) \cdots \big( \sum_{k_n} s^{k_n} a^{k_n}_n b^{k_n}_{l_n} \big) \\
  & = \sum^n_{l_1 l_2 \cdots l_n = 1} \varepsilon_{l_1 l_2 \cdots l_n} \\
  &\times \bigg[ \sum_{k_1 k_2 \cdots k_n} s^{k_1} s^{k_2} \cdots s^{k_n} a^{k_1}_1 b^{k_1}_{l_1}  \cdots a^{k_n}_n b^{k_n}_{l_n}  \bigg] \\
  & = \sum_{k_1 k_2 \cdots k_n} s^{k_1} s^{k_2} \cdots s^{k_n}  a^{k_1}_1 \cdots a^{k_n}_n \\
  &\times \bigg[  \sum^n_{l_1 l_2 \cdots l_n = 1} \varepsilon_{l_1 l_2 \cdots l_n}  b^{k_1}_{l_1} \cdots b^{k_n}_{l_n} \bigg]
\end{split}
\label{eq:svd_in}
\end{align}
Note that the inner summation with respect to $l_\nu$ for a specific $\{ k_\mu \}$ ($\mu, \nu = 1, 2, \cdots, n$) gives rise to a determinant, 
which can be denoted as $D^b_{[k_\mu]}$:
\begin{align}
\begin{split}
 D^b_{[k_\mu]} \equiv \sum^n_{l_1 l_2 \cdots l_n = 1} \varepsilon_{l_1 l_2 \cdots l_n}  b^{k_1}_{l_1} \cdots b^{k_n}_{l_n}
 \end{split}
\end{align}
This determinant corresponds to the sub-matrix $b^k_l$ formed by row $1, 2, \cdots, n$ and column $k_1, k_2, \cdots, k_n$.

In the outer summation of Eq. [\ref{eq:svd_in}], each index $k_\mu$ can range from $1$ to the number of singular values.
$D^b_{[k_\mu]}$ is non-zero only when $k_1, k_2, \cdots, k_n$ are not equal to each other, thus placing constraints on values of $k_\mu$ in the outer summation.

Next, we may consider the case where the $n$ values of $k_\mu$ are taken from $1, 2, \cdots, n$, without loss of generality.
Note that there is no ordering presumed in the outer summation of $k_\mu$, therefore summing over all $k_\mu$'s will
generate $n!$ permutations of $1, 2, \cdots, n$.
For these $n!$ permutations, the corresponding $D^b_{[k_\mu]}$ will have the same absolute value, and its $\pm$ sign depends on whether the permutation is odd or even,
due to the nature of determinant.
If we enforce ordering $k_1 < k_2 < \cdots < k_n$, 
we may use the Levi-Civita symbol to represent the sign due to permutation.
\begin{align}
\begin{split}
D  & = \sum_{k_1 k_2 \cdots k_n} s^{k_1} s^{k_2} \cdots s^{k_n}  a^{k_1}_1 \cdots a^{k_n}_n D^b_{[k_\mu]} \\
& =  \sum_{k_1 < k_2 <  \cdots < k_n} s^{k_1} s^{k_2} \cdots s^{k_n}  \\
& \times \sum_{l_1 l_2 \cdots l_n } a^{l_1}_1 \cdots a^{l_n}_n \varepsilon_{l_1 l_2 \cdots l_n} D^b_{[k_\mu]}
\end{split}
\end{align}
where the tuple $l_1 l_2 \cdots l_n$ iterate over all permutation of $k_1 < k_2 < \cdots < k_n$. 
Therefore,
\begin{align}
\begin{split}
D=\sum_{k_1 < k_2 < \cdots < k_n} s^{k_1} s^{k_2} \cdots s^{k_n} D^a_{[k_\mu]} D^b_{[k_\mu]}
\end{split}
\end{align}
\end{proof}

\begin{theorem} 
The absolute value of the determinant of an $n \times n$ sub-matrix of $\bm{\zeta}$ is bound by the product of its largest $n$ singular values provided by SVD
\begin{align}
\begin{split}
|D^{n\times n}| \leq \sum_{k_1 < k_2 < \cdots < k_n} s^{k_1} s^{k_2} \cdots s^{k_n} 
\end{split}
\end{align}
where the singular values satisfy $s^1 \geq s^2 \geq \cdots \geq s^n > 0$.
\label{theorem2}
\end{theorem}

\begin{proof}
We may start by proving the absolute value of the determinant $D^a_{[k_\mu]}$ ($D^b_{[k_\mu]}$) is bound by $1$.
Without loss of generality, we may again assume $i_1 = 1, i_2 = 2, \cdots, i_n = n$ and $k_1 = 1, k_2 = 2, \cdots, k_n = n$.
According to its definition:
\begin{align}
\begin{split}
D^a_{[k_\mu]} = 
\begin{vmatrix} 
a^{1}_1 & a^1_2 &\cdots &  a^1_n \\
a^{1}_1 & a^2_2 &\cdots &  a^2_n \\
\vdots & \vdots & \ddots & \vdots \\
a^{n}_1 & a^n_2 & \cdots & a^n_n \\
\end{vmatrix}
\end{split}
\end{align}
Each row of the above determinant is a vector $\bm{a}^k_{1\times n} = (a^k_1, a^k_2, \cdots, a^k_n)$, 
which is a truncation of a full $1\times N$ vector $\bm{a}^k = (a^i_1, a^i_2, \cdots, a^i_N)$. 
$N$ is the first dimension of $\bm{\zeta}$ and $n \leq N$.
Since the matrix $a^k_i$ is obtained from SVD of $\bm{\zeta}$, each of its row vector is normalized to 1:
$|\bm{a}^k| = 1$, and therefore:
\begin{align}
\begin{split}
|\bm{a}^k_{1\times n} | \leq 1 
\end{split}
\end{align}
The geometric meaning of the determinant $D^a_{[k_\mu]}$ is the volume of a parallelepiped spanned by $n$ vectors $\bm{a}^k$,
where $k = 1, 2, \cdots, n$. 
Since the length of each of its edge $|\bm{a}^k_{1\times n} | \leq 1 $, then the volume of the parallelepiped will be no larger than 1, and thus $|D^a_{[k_\mu]} | \leq 1$ ($|D^b_{[k_\mu]} | \leq 1$).

Using the conclusion of theorem \ref{theorem1},
\begin{align}
\begin{split}
|D^{n \times n }| 
& \leq \sum_{k_1 < k_2 < \cdots < k_n} \big| s^{k_1} s^{k_2} \cdots s^{k_n} D^a_{[k_\mu]} D^b_{[k_\mu]} \big| \\
& \leq \sum_{k_1 < k_2 < \cdots < k_n} s^{k_1} s^{k_2} \cdots s^{k_n}
\end{split}
\end{align}
Note that all the singular values of $\zeta$ are non-negative.
\end{proof}

\subsection{An analysis of $\zeta$-matrices with singular-value decomposition}
\unskip
We analyze several representative examples here to illustrate the usefulness of the two theorems for $\bm{\zeta}$ matrices:
a 1D single-atom chain at half filling, the \ce{C} $K$ edge of graphite, the \ce{Cu} $K$ edge copper, the \ce{Li} $K$ edge of lithium metal, 
and the \ce{O} $K$ edge of rutile \ce{TiO2},  \ce{CrO2}, \ce{RuO3}, and \ce{LiCoO2}.
The crystal structures of \ce{RuO3}, and \ce{LiCoO2} are shown in Fig. \ref{structures}.
Both systems are layered structures that will allow lithium insertion/removal, and are being studied as cathode prototypes of rechargeable batteries.
The chosen systems are gapless except for \ce{TiO2} and \ce{LiCoO2}, which have a DFT-PBE band gap of $2.1$ and $1.8$ eV (with a Hubbard $U$ value of $3.3$ eV on the \ce{Co} atom) respectively. 

\begin{figure}[H]
\centering
\includegraphics[width=0.95\linewidth]{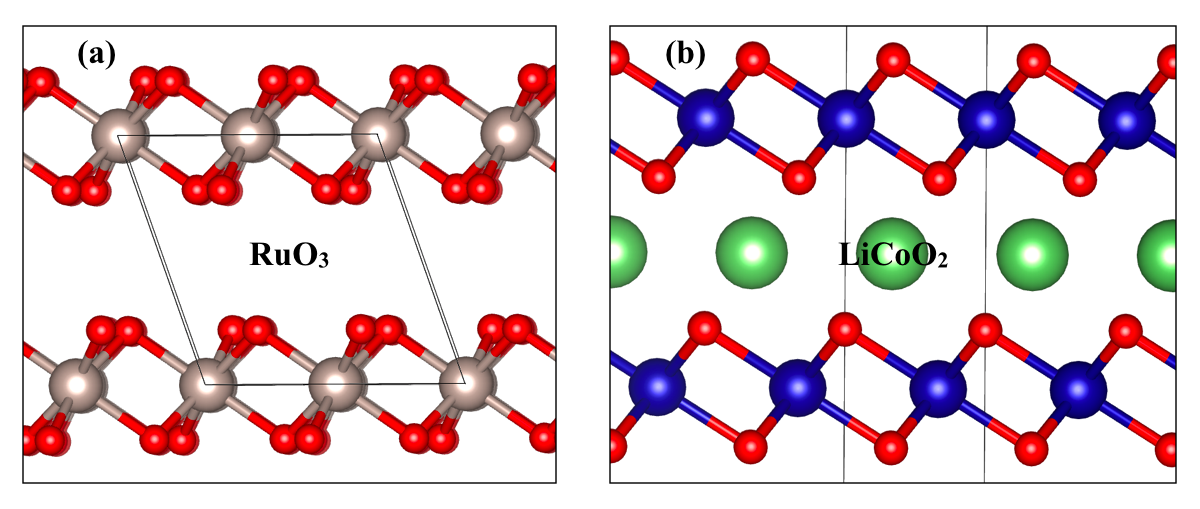}
\caption{layered structures (side view) of \ce{RuO3} (\textbf{a}) and \ce{LiCoO2} (\textbf{b}). The unit cells are outlined by the thin black lines.}
\label{structures}
\end{figure}   

The tight-binding model for the 1D chain reads:
\begin{align}
\begin{split}
H = -t \sum^N_{j = 1} \big( c^\dagger_{j + 1} c_j +  c^\dagger_{j } c_{j + 1}  \big)
\end{split}
\end{align}
where a periodic boundary condition is employed and each site has double occupancy.
The above model is considered as the initial-state Hamiltonian $H_i$.
The core-hole potential is assumed to act on a single site (the site at $j = 0$): $V_c = \Delta V c^\dagger_0 c_0$, thus the final-state Hamiltonian is $H_f = H_i + V_c$.
It can be solved that the 1p wavefunctions for $H_i$ are:
\begin{align}
\begin{split}
|k \sigma \rangle =  \frac{1}{\sqrt{N}}\sum^N_{j = 0} e^{i k j  } | j \sigma \rangle
\end{split}
\end{align}
where $k = 0, \frac{1}{N}, \cdots, \frac{N - 1}{N}$ and $\sigma = \uparrow, \downarrow$. 
We then define the 1p XAS matrix element as: 
\begin{align}
\begin{split}
\langle k\sigma | j = 0, \sigma \rangle =  \frac{1}{\sqrt{N}}
\end{split}
\end{align}

In the actual calculation, we choose the number of sites $N = 200$, the number of electrons $N_e = 200$ (half-filled), $t = 1$, 
a perturbation potential of $\Delta V = -100$ at the excited site to simulate the core hole effect (set to 100 for exaggeration).
We find the determinant spectrum barely changes after $\Delta V < -100$. 

\begin{figure}
\centering
\includegraphics[width=0.95\linewidth]{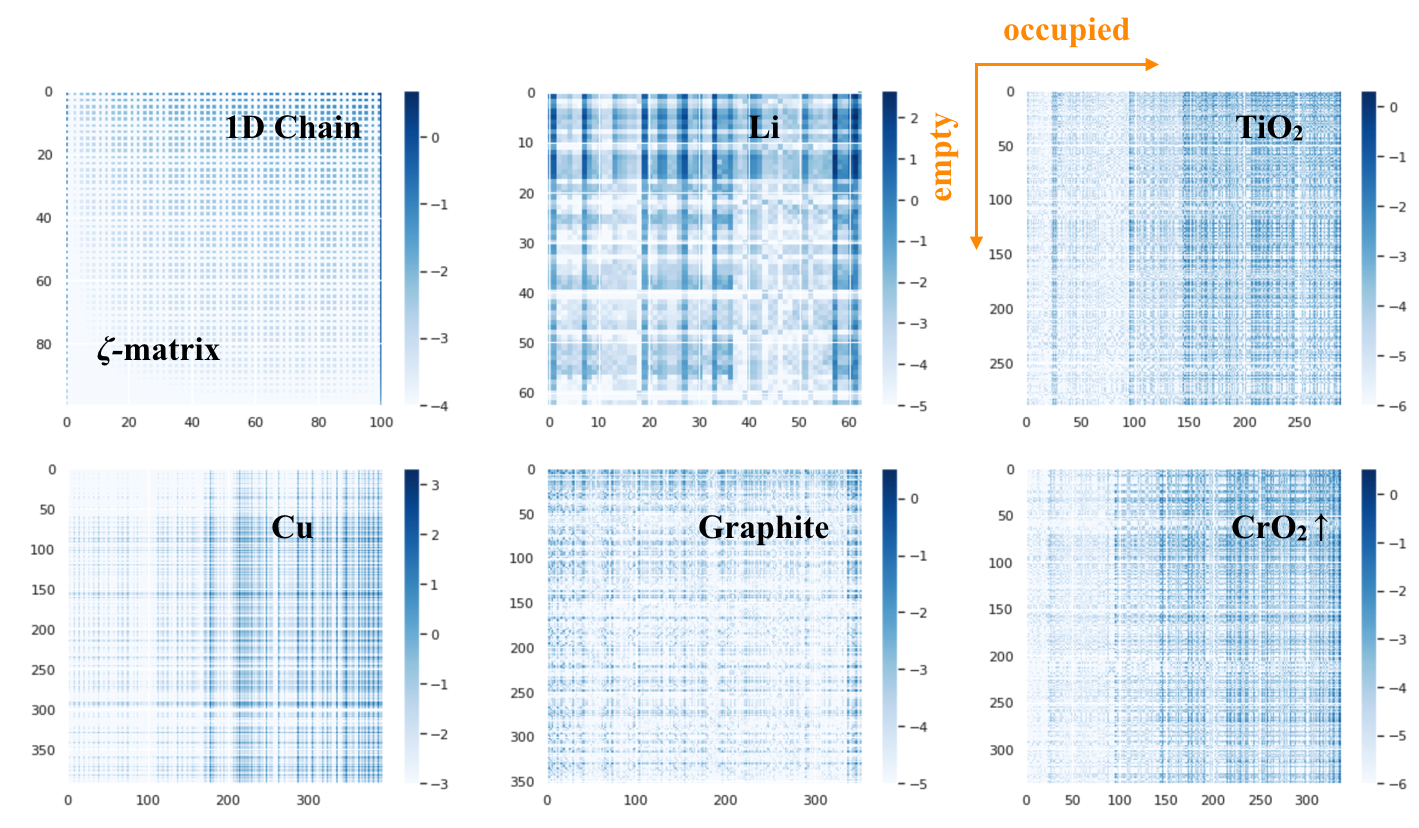}
\caption{Heat-map plot of $log_{10} |\zeta_{ij}|$ of the XAS $\zeta$ matrices of several chosen systems for illustrating the cross-stripe pattern.
If not plot in logarithm scale, the $\zeta$ matrices will appear to be sparse.
The rows correspond to empty orbitals and the columns to occupied orbitals. 
The number of occupied orbitals in the 1D Chain, Li, \ce{TiO_2}, Cu, Graphite, and \ce{CrO_2} (the spin-up channel) supercell are $100$, $64$, $288$, $352$, $392$ and $336$ respectively, 
which determines the number of columns. }
\label{zeta}
\end{figure} 

First, we plot the $\zeta$ matrices for the XAS (\ce{Li}, \ce{C}, \ce{O}, and \ce{Cu} $K$ edges) of six chosen systems in order to exemplify the aforementioned stripe pattern.
Once we have completed the $\Delta$SCF calculation for both the initial and final state, the $\zeta$ matrices can be calculated using Eq. [\ref{eq:zeta_def}]. 
Several $\zeta$ matrices are shown in Fig. \ref{zeta}. 
Although the chosen systems are vastly different in terms of crystal symmetry and bonding nature, it can be seen that all the plotted $\zeta$ matrices display 
a cross-strip pattern that runs both vertically and horizontally.
This cross-strip pattern suggested that the $\zeta$ matrix for x-ray excitations can be universally expressed in terms of Kronecker product of two vectors: $\zeta_{ij} \sim a_i b_j$, with a few small residual terms.
Thus it becomes natural to introduce an SVD analysis to the $\zeta$ matrix, and apply the two theorems proved above to determine the spectral convergence with respect to shakeup order.

\begin{figure*}
\centering
\includegraphics[width=0.95\linewidth]{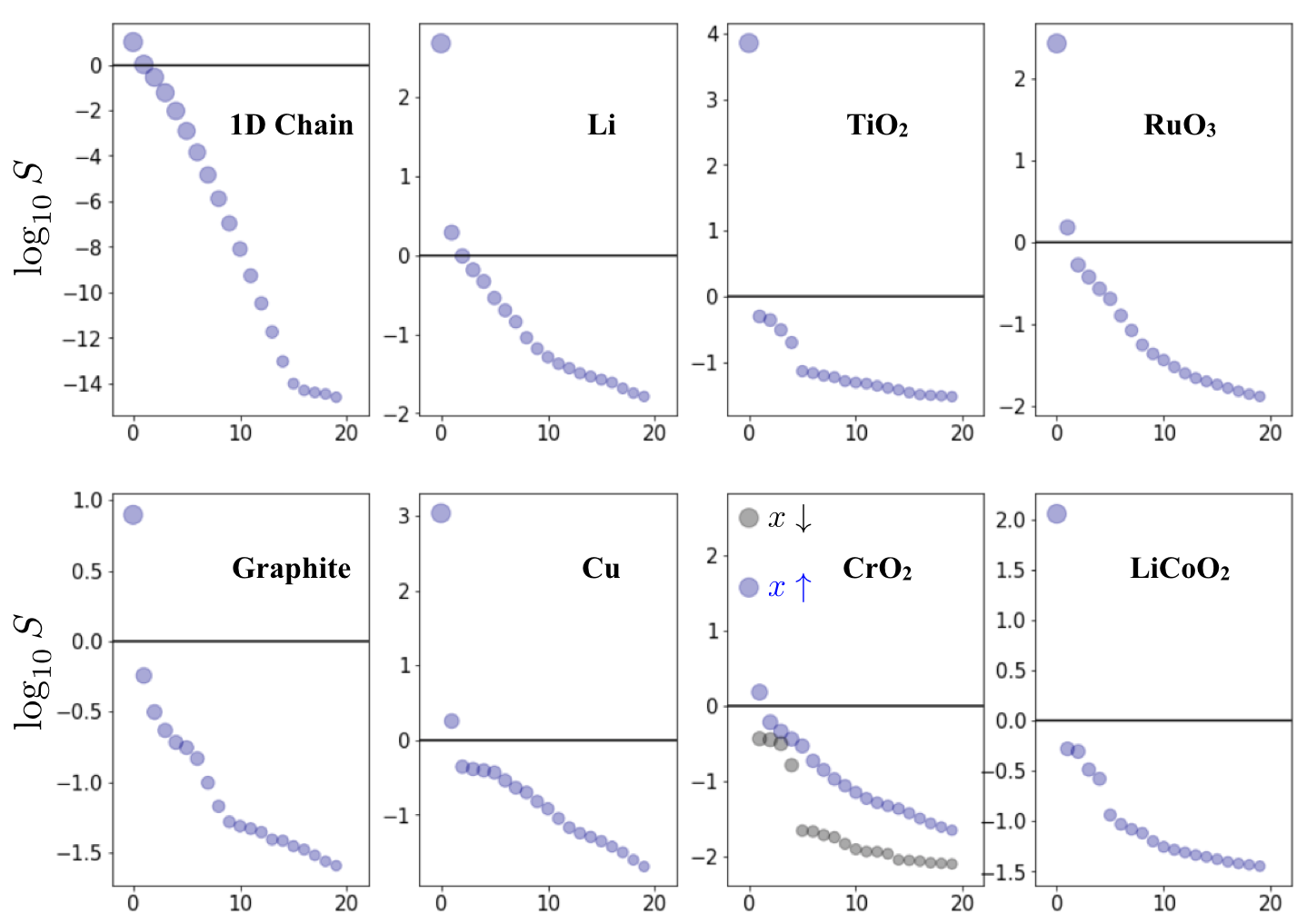}
\caption{20 most significant singular values of the $\zeta$ matrices. 
The singular values are positive definite and scaled logarithmically.
The horizontal line at $S = 1$ is used to divide values larger than 1 and smaller than 1.
For \ce{CrO2}, two series of singular values for spin-wise $\zeta$ matrices are plotted.}
\label{svd}
\end{figure*}   

The twenty most significant singular values of the $\zeta$ matrices of the chosen systems are shown in Fig. \ref{svd} (the singular values are plotted logarithmically). 
The first 5 largest singular values (denoted as $s^i$) are also tabulated in Table. \ref{tab:svd}.
It can immediately be seen for all the $\zeta$ matrices that only a few singular values are larger than $1$, 
especially for graphite, \ce{TiO2}, \ce{LiCoO2} and the insulating (spin-down) channel of \ce{CrO2} with only 1 $s^i > 1.0$.
The latter three are systems with a large band gap.
The other systems only have two singular values larger than $1$.
Note that in \ce{Li}, $s^3 = 0.977$ and $s^4 = 0.653$, which can be rounded to 1.
In all the cases, $s^i < 0.1$ for $i > 10$, which are orders of magnitude smaller than $s^1$.
This means only a few singular values $s^i$ are relevant for the analysis of shakeup orders.

\setlength\extrarowheight{5pt}
\setlength{\tabcolsep}{5pt}
\begin{table*}
\caption{The largest five singular values of the zeta matrices of the studies systems, and their accumulative products.}
\centering
\begin{tabular}{r|rrrrr|rrrr}
\toprule
\textbf{System}	& \textbf{$s^1$}	& \textbf{$s^2$} & \textbf{$s^3$}	& \textbf{$s^4$} & \textbf{$s^5$} & \textbf{$s^1 s^2$} & \textbf{$s^1 s^2 s^3$} & \textbf{$s^1 s^2 s^3 s^4$} & \textbf{$s^1 s^2 s^3 s^4 s^5$} \\
\hline
1D Chain		& 9.9683 &    1.0287 &    0.2810 &    0.0592 &    0.0094		&   10.2545 &    2.8820 &    0.1707 &    0.0016\\
Li			&  478.8850 &    1.9354 &    0.9774 &    0.6528 &    0.4657	&  926.8103 &  905.8236 &  591.2829 &  275.3625\\
\ce{TiO2}		& 7092.6499 &    0.5008 &    0.4417 &    0.3118 &    0.2006		& 3552.0365 & 1569.0875 &  489.1961 &   98.1572\\
\ce{RuO3} 	&  267.8267 &    1.5267 &    0.5323 &    0.3790 &    0.2721	&  408.8852 &  217.6645 &   82.4916 &   22.4428   \\
Graphite		&    7.8617 &    0.5691 &    0.3137 &    0.2325 &    0.1913 &    4.4744 &    1.4035 &    0.3264 &    0.0624 \\
\ce{Cu}		& 1076.8833 &    1.7903 &    0.4386 &    0.4087 &    0.3919. & 1927.8910 &  845.5359 &  345.5765 &  135.4185 \\
\ce{CrO2}$\uparrow$ &   37.0903 &    1.5178 &    0.6071 &    0.4586 &    0.3631 &   56.2962 &   34.1768 &   15.6735 &    5.6917	\\
\ce{CrO2}$\downarrow$ &  310.3019 &    0.3667 &    0.3556 &    0.3157 &    0.1634 &  113.7812 &   40.4625 &   12.7736 &    2.0867	\\
\ce{LiCoO2}	&  113.3031 &    0.5227 &    0.4932 &    0.3257 &    0.2632 &   59.2231 &   29.2065 &    9.5120 &    2.5036	\\ 
\hline
\end{tabular}
\label{tab:svd}
\end{table*}

With these numerical results and the properties of $\{s^i\}$, we can revisit the two theorems and discuss the contributions from different shake-up order $f^{(n)}$ to the x-ray spectra.
Because the vast majority of $s^i$ are small, we could take the product of the leading $s^i$ to form an analysis of order of approximation.
As indicated in the Theorem \ref{theorem1}, the $n \times n $ sub-determinants corresponding to the $f^{(n)}$ terms are associated with the coefficients $s^{k_1} s^{k_2} \cdots s^{k_n}$.
We may thus take the largest $n$ singular values $s^1$ and calculate their product to estimate the contribution of $f^{(n)}$:
\begin{align}
\begin{split}
f^{(1)}:& s^1 \\
f^{(2)}:& s^1 s^2 \\
f^{(3)}:& s^1 s^2 s^3 \\
\cdots \\
\end{split}
\end{align}
The accumulative product $s^1 s^2 \cdots s^n$ of some leading-order singular values are also shown in Table. \ref{tab:svd}.
For \ce{Li}, \ce{RuO3}, \ce{Cu}, \ce{CrO2}$\uparrow$, $s^1 s^2$ is significantly larger than $s^1$ because $s^2 > 1.5$.
This suggests the contribution of the $f^{(2)}$ terms is of the same order of magnitude of $f^{(1)}$.
An extreme case is \ce{Li} in which even $s^1 s^2 s^3 s^4 \approx 591$ is larger to $s^1 \approx 479$.
This requires one to go beyond $f^{(1)}$ in the XAS determinant calculation for these systems.

\subsection{XAS calculated by the determinant formalism}

To test how good this empirical estimate is, we perform determinant calculations for XAS of the chosen systems, at least at the order of $f^{(2)}$.
The obtained spectra decomposed by different shake-up orders are shown in Fig. \ref{spectra} (odd rows).
We find it is indeed true that the contribution from the $f^{(2)}$ terms is as significant as $f^{(1)}$ in \ce{Li}, \ce{RuO3}, \ce{Cu}, \ce{CrO2}$\uparrow$ (as studied in Ref. \cite{liang2017accurate, liang2018quantum}).
In every case, the contribution from $f^{(2)}$ constitutes more than $40\%$ of the entire computed spectrum.
In particular for \ce{Li}, a trend of convergence is only seen after including the $f^{(4)}$ terms.
The contributions from $f^{(1)}$, $f^{(2)}$, $f^{(3)}$, $f^{(4)}$, and $f^{(5)}$ to the full spectrum ($f^{(1)}+f^{(2)}+f^{(3)}+f^{(4)}+f^{(5)}$) are $29.7 \%$, $51.3 \%$, $14.8 \%$, $3.7 \%$, and $0.5 \%$.
This is consistent with the fact that $s^1 s^2 \approx 927$,  $s^1 s^2 s^3 \approx 906$, and $s^1 s^2 s^3 s^4 \approx 591$ are comparable to $s^1 \approx 479$ in this system.
But note that the contribution $f^{(n)}$ is not entirely proportional to the accumulative product $s^1 s^2 \cdots s^n$.
It is observed that $f^{(n)}$ decays more quickly than $s^1 s^2 \cdots s^n$ as $n$ increases.
This is because we haven't taken into account the decomposed determinants $D^a_{[k_\mu]}$ and $D^b_{[k_\mu]}$ in this empirical estimate, and $|D^a_{[k_\mu]}|, |D^b_{[k_\mu]}| < 1$ holds strictly as proved in Theorem \ref{theorem2}. 

In other investigated systems such as graphite, \ce{LiCoO2}, \ce{TiO2}, and \ce{CrO2}$\downarrow$, 
the contribution from the $f^{(2)}$ terms is noticeably less significant, constituting less than $30\%$ of the full spectrum.
Their corresponding second largest singular values $s^2 < 0.6$, suggesting that $s^1 s^2 < s^1$ and hence $f^{(2)}$ would not be as important as the $f^{(1)}$ terms. 
One severe deviation from this estimate would be the tight-binding 1D chain.
The corresponding $s^1 s^2 \approx 10.25$, which is comparable to $s^1 \approx 9.97$.
However, the $f^{(2)}$ contribution is tiny, about $6.1\%$ of the full ($f^{(1)} + f^{(2)}$) spectrum.
Again, this is because the decomposed determinants $D^a_{[k_\mu]}$ and $D^b_{[k_\mu]}$ are missing.
In this regard, it would be better to view the accumulative product $s^1 s^2 \cdots s^n$ as the \emph{upper bound} of the $f^{(n)}$ contribution.
In other words, if $s^1 s^2 \cdots s^n$ is significantly small compared to the leading order terms ($s^1$ and $s^1 s^2$), then it is already safe to neglect the $f^{(n)}$ terms.

\begin{figure*}
\centering
\includegraphics[width=0.95\linewidth]{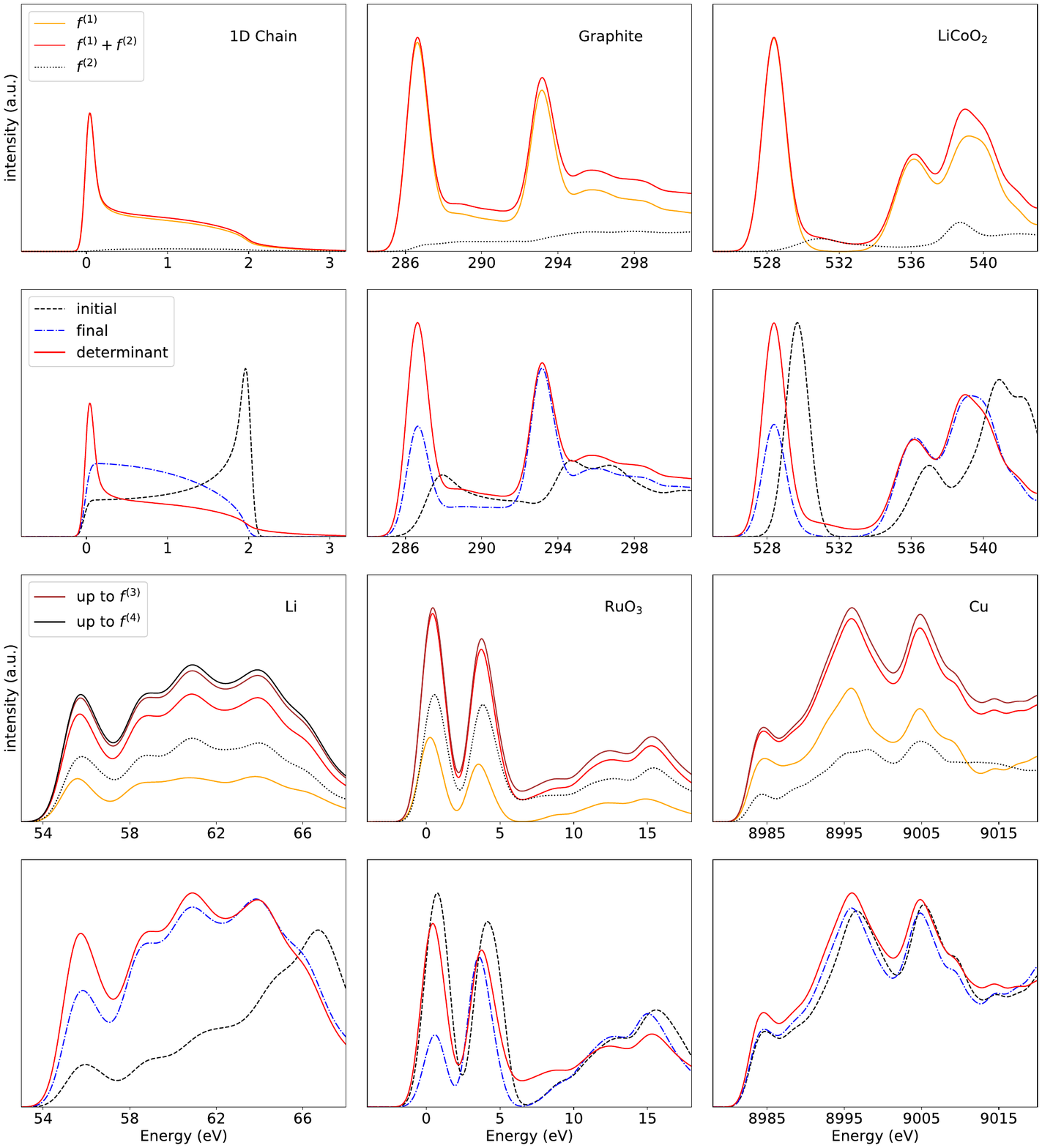}
\caption{(color online) \textbf{Odd rows}: Comparison of determinant spectra at different shake-up orders. \textbf{Even rows}: comparison of determinant spectra and initial- and final-state spectra. 
A high-order determinant spectrum is not plotted if it appears to overlap on the determinant spectrum at one lower order (for example, the \ce{Li} spectrum up to $f^{(5)}$ appear to overlap with the one up to $f^{(4)}$). 
The determinant spectra are scaled for comparing to the initial- and final-state spectra, 
according to the absorption peak at around 293 eV for graphite, 536 eV for \ce{LiCoO2}, 64 eV for \ce{Li}, 4 eV for \ce{RuO3}, and 9005 eV for \ce{Cu}.
}
\label{spectra}
\end{figure*} 

To define the smallness of the cumulative product, we may introduce the ratio $\eta_n$ such that: $\eta_n \max\limits_m \{s^1 s^2 \cdots s^m\} = s^1 s^2 \cdots s^n $,
where $\max\limits_m \{s^1 s^2 \cdots s^m\}$ is the maximum of the cumulative product.
According to the examples studied in this work, we find one could safely neglect the $f^{(n)}$ configurations with $\eta_n < 0.5$.

Although we have gone beyond $f^{(2)}$ configurations in the determinant calculation, 
\emph{we still find the spectra for all the cases studied in this work can be well defined by just the configurations up to $f^{(2)}$}.
This applies even to \ce{Li}, in which up to $f^{(5)}$ configurations are included.
The $f^{(2)}$ configurations slightly increase intensities of the two absorption humps at $61$ and $64$ eV compared to the absorption peak at $56$ eV.
And the spectra beyond $f^{(2)}$ (summed from $f^{(1)}$ to $f^{(n)}$ for $n = 3, 4, 5$) can basically be reproduced by scaling up the $f^{(1)} + f^{(2)}$ spectrum using an energy-independent factor.

Here, we take a closer examination of the intensity contribution from $f^{(2)}$ configurations (as shown by black dotted curves in Fig. \ref{spectra}).
The $f^{(2)}$ contributions in the 1D chain, graphite, and \ce{LiCoO2} (the first row of Fig. \ref{spectra}) are different from the ones as in \ce{Li}, \ce{RuO3}, and \ce{Cu} (the third row of Fig. \ref{spectra}). 
In the former group, the $f^{(2)}$ contribution does not have as many absorption features compared to $f^{(1)}$.
In particular, the $f^{(2)}$ contribution does not have a well-defined peak at the absorption edge as manifests in $f^{(1)}$. 
In the latter group,  the $f^{(2)}$ contribution mimics the $f^{(1)}$ contribution in that they have similar absorption structures (peaks).
All the absorption peaks that appear in  $f^{(1)}$ also appear in $f^{(2)}$, including the peaks at the absorption onset.
A primary reason for the $f^{(2)}$ contribution mimicking the $f^{(1)}$ one in \ce{Li}, \ce{RuO3}, and \ce{Cu} is due to low-energy \emph{e-h} pair production in these metallic systems.
There are many orbitals of similar energies near the Fermi level in the final-state system (the supercell that contains an impurity for modeling core-excited states).
More speficially, there are more than 12 orbitals within an energy window of 0.3 eV at the $\Gamma$-point of the supercell, as found by the $\Delta$SCF calculations.
Some of these orbitals are occupied and some are not.
The majority of the $f^{(2)}$ configurations in these systems can be understood as a transition from the core level to one empty final-state orbital $c$, 
coupled with a low-energy \emph{e-h} pair (as labeled by $c_1$ and $v_1$), whose excitation energy is negligibly small.
Therefore the spectral contribution from $f^{(2)}$ mimics the one from $f^{(1)}$, which is only defined by empty final-state orbitals, in a one-body manner.
\emph{This partially explains why the one-body final-state method can often reproduce the XAS of many metallic systems adequately,  even though the physics of MND theory is completely missing}.

The $f^{(2)}$ contributions in the 1D chain, graphite, and \ce{LiCoO2} are more complex to analyze, 
which involves not just low-energy \emph{e-h} pairs near the Fermi level / band gap.
In graphite, for instance, many $f^{(2)}$ configurations that contribute modestly to XAS are composed of a core-level transition to $c$ which is coupled to an \emph{e-h} pair with $c_1$ being a low-lying empty orbital ($\pi^*$), and $v_1$ going over the continuum of occupied orbitals that span more than $8$ eV ($\pi$ and $\sigma$ continuum).
So the transition energy should be added from the two \emph{e-h} pairs:
$E = (\varepsilon_c - \varepsilon_h) + (\varepsilon_{c_1} - \varepsilon_{v_1})$, 
in which both $\varepsilon_c - \varepsilon_h$ and $\varepsilon_{c_1} - \varepsilon_{v_1}$ can vary across a wide energy range.
This explains why the $f^{(2)}$ spectral contribution in graphite is smeared out without well-defined peaks.
Similar analysis applies to the $f^{(2)}$ continuum between $530$ and $534$ eV in \ce{LiCoO2}, 
which involves transitions from the core level and some valence orbitals $v_1$ (mainly \ce{O} $2p$) around $4$ eV below the valence band maximum (VBM)
to two low-lying empty orbitals $c$ and $c_1$.
This $f^{(2)}$ continuum fills out the energy gap between $530$ and $534$ eV, 
and explains why there is no spectral gap between near-edge peaks ( $< 6$ eV from onset) and the high-energy humps ($> 6$ eV from onset)
in the \ce{O} $K$ edge of TMOs, 
although there is no single-body orbital within this gap.
Such is true for the energy gap near $6$ eV in \ce{RuO3}.

It should be noted that configuration interaction is absent from the above analysis of $f^{(2)}$ configurations.
Interaction within $f^{(2)}$ may introduce excitonic effects between $c_1$ and $v_1$ and plasmon excitations.
So far there is no electron-plasmon (plasmaron) coupling present in the determinant approach.
Configuration interaction between $f^{(1)}$ and $f^{(2)}$ may remix the spectral contribution from the two sets of configurations and modify the shakeup effects. 
So far the ratio of the intensities of the near-edge peaks to the high-energy humps in TMOs is still too high in the determinant calculation up to $f^{(2)}$, as compared with experiments.
Introducing configuration interaction may resolve this problem, which is beyond the scope of this work.

\subsection{Initial-state, final-state, and the determinant XAS}

Lastly, we rationalize the determinant calculation by comparing the determinant spectra with the initial- and final-state spectra (even rows in Fig. \ref{spectra}).

For the 1D tight-binding chain, the MND effects manifest as an asymmetric, power-law singularity that diverges at the absorption edge.  
Note that this singularity has already been reproduced at the $f^{(1)}$ level.
Thanks to the final-state orbitals in the determinant formalism, 
one does not need to go over many orders of Feynman diagrams expanded in the initial-state orbitals to produce the edge singularity.
The determinant spectrum does not resemble to the one-body final-state spectrum, which looks like a "half dome", 
making the inclusion of MND effects essential to the spectral calculation.
The singularity at $2$ eV is the van Hove singularity due to the 1D band edge.

Similar MND effects also manifest in the XAS (polarization vector is $45\deg$ off-plane) of graphite.
After the correction of the determinant approach, the first-peak (around $286$ eV, due to $\pi^*$) intensity is significantly magnified compared with the final-state XAS.
The corrected intensity ratio of the first peak to the second (around $293$ eV) is $1.25$ ($0.65$ in the final-state spectrum), which is in good agreement with $1.35$ in a previous experiment
\cite{weser2010induced}. 
The spectral plateau between $288$ and $292$ eV (due to the constant joint DOS in 2D systems) is also tilted upward at the low-energy end, due to the MND effects.
Such MND effects in graphite were also obtained from first-principles using a more complex approach based on Green's function, 
which involves energy and time integral, and Fourier transformation \cite{wessely2005ab, wessely2006dynamical, ovcharenko2013specific}.
Here, the determinant approach provides equivalent spectra by only a one-shot matrix inversion and a heuristic search of computationally accessible configuration space.

Significant intensity correction of the peak at onset is also observed in \ce{LiCoO2}, \ce{Li}, and \ce{RuO3}.
In \ce{LiCoO2}, the intensity ratio of the peak at $528$ eV to the peak at $532$ eV as found by the determinant approach is $2:1$, 
which is in good agreement with a previous measurement \cite{juhin2010angular}.
In the same work, the intensity ratio by the one-body final-state approach with GGA + U is $1:1$, which is consistent with our calculation. 
In \ce{RuO3}, the intensity ratios of two near-edge peaks are reversed after correction, which was also reported in our previous work on 3d TMOs \cite{liang2017accurate, liang2018quantum}.

The Li $K$ edge calculated from the determinant formalism (at the level of $f^{(4)}$) for \ce{Li} metal is in good agreement with a previous experiment \cite{miedema2014in_situ}.
Thermal vibration of the lattice could further modify the spectrum \cite{pascal2014finite}, and may broaden the peaks at 61 and 63 eV, making it in closers agreement with the experiment.
However, thermal effects will not be discussed in this work, because we focus on the convergence with respect to shakeup orders.

The core-hole attraction effect is most significant in the 1D Chain, graphite, and \ce{Li}, although they are gapless systems.
After the inclusion of the core hole, the initial-state spectra dramatically redshift.
However, the core-hole effect only causes the initial-state spectra to redshift rigidly in \ce{LiCoO2}, \ce{RuO3}, and \ce{Cu}.
The redshifts of the lowest-energy peaks in  \ce{LiCoO2} and \ce{RuO3} are $1.28$ and $0.34$ eV respectively, 
which are not negligibly small, although the core hole is at the \ce{O} site where the near-edge orbitals are mainly composed of TM $d$ orbitals.
This explains why initial-state spectra are sometimes good approximation to XAS for \ce{O} $K$ edges in TMOs, as discussed in Ref. \cite{de2008core}.

\section{Conclusions}

In summary,  we have introduced two theorems for regulating the convergence of the determinant calculation, using a SVD analysis over the auxiliary $\zeta$-matrix.
The convergence with respect to the excitation order $n$ depends on the number of the significant singular values of the $\zeta$ matrix.
We show that the cumulative product of the singular values $s^1 s^2 \cdots s^n$ can be used as an effective estimate for the $f^{(n)}$ contribution.
It is found empirically that it is safe to neglect the $f^{(n)}$ contribution and higher-order when $s^1 s^2 \cdots s^n < 0.5 \max\limits_m \{s^1 s^2 \cdots s^m \}$.
However, satisfactory determinant spectra have been achieved at the order of $f^{(2)}$ for all the examined cases (in this work and the TMOs as in Ref. \cite{liang2018quantum}).

\section*{Acknowledgement} Theoretical and computational work was performed by Y. L. and D. P. at The Molecular Foundry, which is supported by the Office of Science, Office of Basic Energy Sciences, of the United States Department of Energy under Contact No. DE-AC02-05CH11231. Computations were performed with the computing resources at the National Energy Research Scientific Computing Center (NERSC).

\appendix
\section{Computational Details}
The single-body energies and orbitals for both the initial- and final-state systems are obtained with DFT $\Delta$ SCF calculations, as described in Ref. \cite{liang2017accurate, liang2018quantum}.
Modified pseudopotentials (by changing $1s^2$ to $1s^1$) are generated for simulating the core-hole potential of $K$ edges.
Supercell dimensions are chosen to be approximately $10$\AA{} that is sufficient to minimize spurious periodic interactions among the core-hole impurities.
The $\Delta$ SCF calculations for TMOs are performed using the DFT+$U$ theory \cite{dudarev1998electron} with the $U$ value adopted from Ref. \cite{wang2006oxidation}.
A $5\times 5 \times 5$ $\bm{k}$-grid is used to sample the BZ of the supercell, which is essential for high-energy scattering states and metallic systems. 
The DFT part of the calculations is performed using a local repository of the ShirleyXAS code, which is available at the David Prendergast's group at the Molecular Foundry.

The determinant calculations are performed using an open source software package, MBXASPY, which is available at \url{https://github.com/yufengliang/mbxaspy}.

\end{document}